\documentclass[11pt]{article}
\usepackage{amsfonts}
\usepackage{graphics,amssymb,amsthm,amsmath,epsf}
\usepackage{graphicx}
\usepackage{subfigure}

\usepackage[round]{natbib}

\newtheorem{theorem}{Theorem}[section]
\newtheorem{prop}{Proposition}

\newtheorem{corollary}{Corollary}[section]
\newtheorem{remark}{Remark}[section]

\allowdisplaybreaks

\begin{document}
\title{On Bivariate Generalized Exponential-Power Series \\ Class of Distributions}
\author{Ali Akbar Jafari\thanks{Corresponding: aajafari@yazd.ac.ir \ \ Tel:+98-353-8122699 \ \ Fax:+98-353-8210695  } , Rasool Roozegar \ \\
{\small Department of Statistics, Yazd University, Yazd,  Iran}\\}
\date{}
\maketitle \thispagestyle{empty}

\begin{abstract}
In this paper, we introduce a new class of bivariate distributions by compounding the bivariate generalized exponential and power-series  distributions.
 This new class contains some new sub-models such as the bivariate generalized exponential distribution, the bivariate generalized exponential-poisson, -logarithmic, -binomial and -negative binomial distributions. We derive different properties of the new class of distributions. The EM algorithm is used to determine the maximum likelihood estimates of the parameters. We illustrate the usefulness of the new distributions by means of an application to a real data set.
\end{abstract}

\section{Introduction}
The modeling of lifetime is an important aspect of statistical work in a variety of scientific and technological fields. In this area, much new univariate distributions have been studied in the statistical literature recently. Interestingly, not much work has been done on the bivariate distributions mainly due to its analytical intractability.

The two-parameter generalized exponential (GE) distribution has been introduced by
\cite{gu-ku-99}
and it has the following cumulative distribution function (cdf) and probability density function (pdf), respectively:
\begin{eqnarray}
&&F_{GE}(x;\alpha ,\lambda )={(1-e^{-\lambda x})}^{\alpha },\ \ \ \ \ x>0,\ \ \ \alpha ,\lambda >0, \label{eq.FGE}\\
&&f_{GE}(x;\alpha ,\lambda )=\alpha \lambda e^{-\lambda x}{(1-e^{-\lambda x})}^{\alpha -1},\ \ \ \ \ x>0,\ \ \ \alpha ,\lambda >0.
\end{eqnarray}

The hazard function of GE distribution can be increasing, decreasing and constant, but it cannot be bathtub shaped. Therefore,
\cite{ma-ja-12}
introduced  the generalized exponential-power series (GEPS) distributions  by compounding GE distribution with power series class of distributions. The proposed class includes
GE, Poisson-exponential
\citep{ca-lo-fr-ba-11},
complementary exponential-geometric
\citep{lo-ro-ca-11},
and complementary exponential-power series
 \citep{fl-bo-ca-13}
distributions.

Recently,
\cite{ku-gu-09-BGE}
extended the GE distribution to a bivariate distribution. The cdf of the bivariate generalized exponential (BGE) with parameters $\alpha_1$, $\alpha_2$, $\alpha_3$ and $\lambda$ is given by
\begin{equation}
F_{BGE}\left(x_1,x_2;\alpha _1,\alpha_2,\alpha_3,\lambda \right)=
\left\{
\begin{array}{ll}
(1-e^{-\lambda x_1})^{{\alpha }_1+{\alpha }_3}(1-e^{-\lambda x_2})^{{\alpha }_2} & \ \ {\rm if}\ \ \  x_1\le x_2 \\
(1-e^{-\lambda x_1})^{{\alpha }_1}(1-e^{-\lambda x_2})^{{\alpha }_2+{\alpha }_3} & \ \ {\rm if}\ \ \  x_1>x_2.
\end{array}
\right.
\end{equation}

Note that the BGE distribution has both an absolutely continuous part and a singular part similar to the bivariate exponential distribution reported in
\cite{ma-ol-67}
and the bivariate models proposed by
\cite{sa-ba-07-new}.

In this paper, we compound the BGE distribution and power series class of distributions and define a new class of bivariate distributions. It contains the BGE and GEPS distributions and is called the bivariate generalized exponential-power series (BGEPS) distributions. This paper is organized as follows. In section \ref{sec.class}, we introduce the BGEPS model and obtain some properties of this new family. Some special models are studied in detail in Section \ref{sec.spe}. We propose an EM algorithm to estimate the model parameters in Section \ref{sec.est}. A real data application of the BGEPS distributions is illustrated in Section \ref{sec.exa}.

\section{The BGEPS class}

\label{sec.class}
A random variable $N$ follows the power series distribution if it has the following probability mass function
\begin{equation}
P(N=n)=‎\frac{a_n \theta ^{n} }{C(\theta)}, \ \ \ n=1,2,‎‎\ldots‎‎ ,‎
\end{equation}
where $a_n‎\geq 0‎$ depends only on $n$, $C(\theta)=\sum_{n=1}^{\infty} a_n \theta ^{n} $ and $\theta ‎\in (0,s)‎$ ($s$ can be $\infty$) is such that $C(\theta)$ is finite. Table \ref{tab.ps} lists some particular cases of the truncated (at zero) power series distributions (geometric, Poisson, logarithmic, binomial and negative binomial). Detailed properties of power series distribution can be found in
\cite{noack-50}.
Here, $C'(\theta )$, $C''(\theta )$ and $C'''(\theta )$ denote the first, second and third derivatives of $C(\theta)$ with respect to $\theta$, respectively.

\begin{table}

\begin{center}
\caption{Useful quantities for some power series distributions.}\label{tab.ps}
\begin{tabular}{|l| c c c c c c| }\hline
‎ Distribution
& $a_n$ & $C(\theta)$ & $C^{\prime}(‎\theta)$ & $C^{\prime\prime}(‎\theta)$ & $C^{\prime\prime\prime}(‎\theta)$ &   $s$ \\ \hline
Geometric & $1$ & $‎\theta (1-‎\theta)^{-1}$ & $(1-‎\theta)^{-2}$ & $2(1-‎\theta)^{-3}$ & $6(1-‎\theta)^{-4}$ &  $1$ \\
Poisson & $n!^{-1}$ & $e^{‎\theta}-1$ & $e^{‎\theta}$ & $e^{‎\theta}$ & $e^{‎\theta}$ &  $\infty‎$ \\
Logarithmic & $n^{-1}$ & $-\log(1-‎\theta)$ & $(1-‎\theta)^{-1}$ & $(1-‎\theta)^{-2}$ & $2(1-‎\theta)^{-3}$  & $1$ \\
Binomial & $\binom {k} {n}$
& $(1+‎\theta)^k-1$ & $\frac{k}{(\theta+1)^{1-k}}‎$ & $\frac{k(k-1)}{(\theta+1)^{2-k}}‎$ & $\frac{k(k-1)(k-2)}{(\theta+1)^{3-k}}$ &  $\infty$ \\
Negative Binomial & $\binom {n-1} {k-1}$
& $‎\frac{\theta^k}{ (1-‎\theta)^{k}}$ & $‎\frac{k\theta^{k-1}}{ (1-‎\theta)^{k+1}}$ & $‎\frac{k(k+2\theta-1)}{\theta^{2-k} (1-‎\theta)^{k+2}}$&
$‎\frac{k(k^2+6k\theta+6\theta^2-3k-6\theta+2)}{ \theta^{3-k}(1-‎\theta)^{k+3}}$
& $1$  \\ \hline
\end{tabular}

\end{center}
\end{table}

Now, suppose $\{(X_{1n}, X_{2n}); n=1,2\dots \}$ is a sequence of independent and identically distributed (i.i.d.) non-negative bivariate random variables with common joint distribution function $F_{\boldsymbol X}(.,.)$, where ${\boldsymbol X}=(X_1, X_2)'$. Take $N$ to be a power series random variable independent of $(X_{1i}, X_{2i})$. Let
\[Y_i=\max  \{X_{i1},\dots X_{iN}\} ,\ \ \ \ i=1,2.\]
Therefore, for the joint random variables ${\boldsymbol Y}=(Y_1,Y_2)'$ we have
\begin{equation}\label{eq.pdfyn}
P(Y_1\leq y_1,Y_2\leq y_2,N=n)=F_{Y_1,Y_2|N}(y_1, y_2|n)P(N=n)=(F_{{\boldsymbol X}}(y_1, y_2))^n\frac{a_n{\theta }^n}{C(\theta )}.
\end{equation}
Therefore, the joint cdf of ${\boldsymbol Y}=(Y_1,Y_2)'$ becomes
\begin{equation}
F_{{\boldsymbol Y}}(y_1, y_2)=\sum^{\infty }_{n=1}{{(F_{{\boldsymbol X}}(y_1,y_2))}^n\frac{a_n{\theta }^n}{C(\theta )}}=\frac{C(\theta F_{{\boldsymbol X}}y_1, y_2))}{C(\theta )}.
\end{equation}
In this case, we call ${\boldsymbol Y}$ has a bivariate F-power series (BFPS) distribution.

The corresponding marginal distribution function of $Y_i$ is
\[F_{Y_i}(y_i)=\frac{C(\theta F_{{{\boldsymbol X}}_{{\boldsymbol i}}}(y_i))}{C(\theta )},\ \ \ \ \ i=1,2.\]
In recent years many authors have considered this univariate class: for example the GEPS distribution by
\cite{ma-ja-12}
and the complementary exponential-power series distribution by
\cite{fl-bo-ca-13}
among others.

\begin{remark} If we consider $Z_i=\min  \left\{X_{i1},\dots X_{iN}\right\}$, $ i=1,2,$
another class of bivariate distribution is obtained with the following joint cumulative survival function:
\[{\bar{F}}_{{Z_1,Z_2}}(y_1,y_2)=P(Z_1>y_1,Z_2>y_2)=\frac{C(\theta {\bar{F}}_{{\boldsymbol X}}(y_1, y_2))}{C(\theta)},\]
where ${\bar{F}}_{{\boldsymbol X}}(y_1, y_2)=P(X_1>y_1,X_2>y_2)$. Several papers have studied the univariate case of this class: for example the exponential-power series distribution by
\cite{ch-ga-09}
and the Weibutll-power series distribution by
\cite{mo-ba-11}
 among others.
\end{remark}

In this paper we take $F$ to be the bivariate generalized exponential given in
\eqref{eq.FGE}.
 Therefore, we consider the bivariate generalized exponential-power series (BGEPS) class of distributions which is defined by the following cdf:
\begin{eqnarray}\label{eq.FBGEPS}
F_{{\boldsymbol Y}}(y_1,y_2)&=&\left\{ \begin{array}{ll}
\frac{C(\theta {(1-e^{-\lambda y_1})}^{{\alpha }_1+{\alpha }_3}{(1-e^{-\lambda y_2})}^{{\alpha }_2})}{C(\theta )} & \ \ \ {\rm if}\ \ \  y_1\le y_2 \\
\frac{C(\theta {(1-e^{-\lambda y_1})}^{{\alpha }_1}{(1-e^{-\lambda y_2})}^{{\alpha }_2+{\alpha }_3})}{C(\theta )} & \ \ \ {\rm if}\ \ \  y_1>y_2, \end{array}
\right. \nonumber\\
&=&\left\{ \begin{array}{ll}
\frac{C(\theta F_{GE}(y_1;{\alpha }_1+{\alpha }_3,\lambda )F_{GE}(y_2;{\alpha }_2,\lambda ))}{C(\theta )} & \ \ {\rm if}\ \ \  y_1\le y_2 \\
\frac{C(\theta F_{GE}(y_1;{\alpha }_1,\lambda )F_{GE}(y_2;{\alpha }_2+{\alpha }_3,\lambda ))}{C(\theta )} & \ \  {\rm if}\ \ \  y_1>y_2, \end{array}
\right.
\end{eqnarray}
We denote it by ${\rm BGEPS}({\alpha }_1,{\alpha}_2,\alpha_3,\lambda,\theta )$.

\begin{prop} Let $F_{{\boldsymbol Y}}(y_1,y_2)$ be the cdf of BGEPS distributions given in \eqref{eq.FBGEPS}. Then
\[F_{\boldsymbol Y}(y_1,y_2)=\sum^{\infty }_{n=1}{p_nF_{BGE}(y_1, y_2; n{\alpha }_1,n{\alpha }_2,n{\alpha }_3,\lambda),}\]
where $p_n=P(N=n)=\frac{a_n{\theta }^n}{C(\theta)}$.
\end{prop}

\begin{prop}
Let $(Y_1,Y_2)$ follows  ${\rm BGEPS}\left({\alpha }_1,{\alpha }_2,{\alpha }_3,\lambda ,\theta \right)$ distribution. Then\\
1. Each $Y_i$ has a GEPS distribution with parameters ${\alpha }_i+{\alpha }_3$, $\lambda $ and $\theta $. \\
2.The random variable $U={\max  (Y_1,Y_2)}$ has a GEPS distribution with parameters ${\alpha }_1+{\alpha }_2+{\alpha }_3$, $\lambda $ and $\theta $.\\
3. If $C\left(\theta \right)=\theta $, then ${\boldsymbol Y}$ has a BGE distribution with parameters ${\alpha }_1$, ${\alpha }_2$, ${\alpha }_3$ and $\lambda .$\\
4. $P\left(Y_1<Y_2\right)=\frac{{\alpha }_1}{{\alpha }_1+{\alpha }_2+{\alpha }_3}$.
\end{prop}

\begin{theorem} \label{thm.pdf}
Let ${\boldsymbol Y}$ has a ${\rm BGEPS}\left({\alpha }_1,{\alpha }_2,{\alpha }_3,\lambda ,\theta \right)$ distribution. Then the joint pdf of ${\boldsymbol Y}$ is
\begin{equation}\label{eq.fBGEPS}
f_{{\boldsymbol Y}}(y_1,y_2)=\left\{ \begin{array}{ll}
f_1(y_1,y_2) & \ \ {\rm if}\ \ \  0<y_1<y_2 \\
f_2(y_1,y_2) & \ \ {\rm if}\ \ \  0<y_2<y_1 \\
f_0(y)       & \ \ {\rm if}\ \ \  0<y_1=y_2=y, \end{array}
\right.
\end{equation}
where
\begin{eqnarray}
f_1\left(y_1,y_2\right)&=&\frac{\theta }{C\left(\theta \right)}f_{GE}\left(y_1;{\alpha }_1+{\alpha }_3,\lambda \right)f_{GE}\left(y_2;{\alpha }_2,\lambda \right)[\theta F_{GE}(y_1;{\alpha }_1+{\alpha }_3,\lambda )\nonumber\\
&&\times F_{GE}(y_2;{\alpha }_2,\lambda ) C''\left(\theta F_{GE}(y_1;{\alpha }_1+{\alpha }_3,\lambda )F_{GE}(y_2;{\alpha }_2,\lambda )\right)\nonumber\\
&&+C'\left(\theta F_{GE}(y_1;{\alpha }_1+{\alpha }_3,\lambda )F_{GE}(y_2;{\alpha }_2,\lambda )\right)],\label{eq.f1}\\
f_2\left(y_1,y_2\right)&=&\frac{\theta }{C\left(\theta \right)}f_{GE}\left(y_1;{\alpha }_1,\lambda \right)f_{GE}\left(y_2;{\alpha }_2+{\alpha }_3,\lambda \right)[\theta F_{GE}(y_1;{\alpha }_1,\lambda )\nonumber\\
&&\times F_{GE}(y_2;{\alpha }_2+{\alpha }_3,\lambda )C''\left(\theta F_{GE}(y_1;{\alpha }_1,\lambda )F_{GE}(y_2;{\alpha }_2+{\alpha }_3,\lambda )\right)\nonumber\\
&&+C'\left(\theta F_{GE}(y_1;{\alpha }_1,\lambda )F_{GE}(y_2;{\alpha }_2+{\alpha }_3,\lambda )\right)],\label{eq.f2}\\
f_0\left(y\right)&=&\frac{\theta {\alpha }_3}{C\left(\theta \right)\left({\alpha }_1+{\alpha }_2+{\alpha }_3\right)}f_{GE}\left(y;{\alpha }_1+{\alpha }_2+{\alpha }_3,\lambda \right)C'\left(\theta F_{GE}\left(y;{\alpha }_1+{\alpha }_2+{\alpha }_3,\lambda \right)\right).\nonumber\\\label{eq.f0}
\end{eqnarray}

\end{theorem}

\begin{proof}
It is obvious.
\end{proof}

As a special case, consider $C(\theta)=\theta +{\theta }^{20}$. It is also considered by
\cite{ma-ja-12}.
For $\lambda =1$ and other values of the parameters, the pdf of the BGEPS class of distributions are depicted in Figure \ref{fig.denh}.

\begin{figure}[ht]
\centering
\includegraphics[scale=0.7]{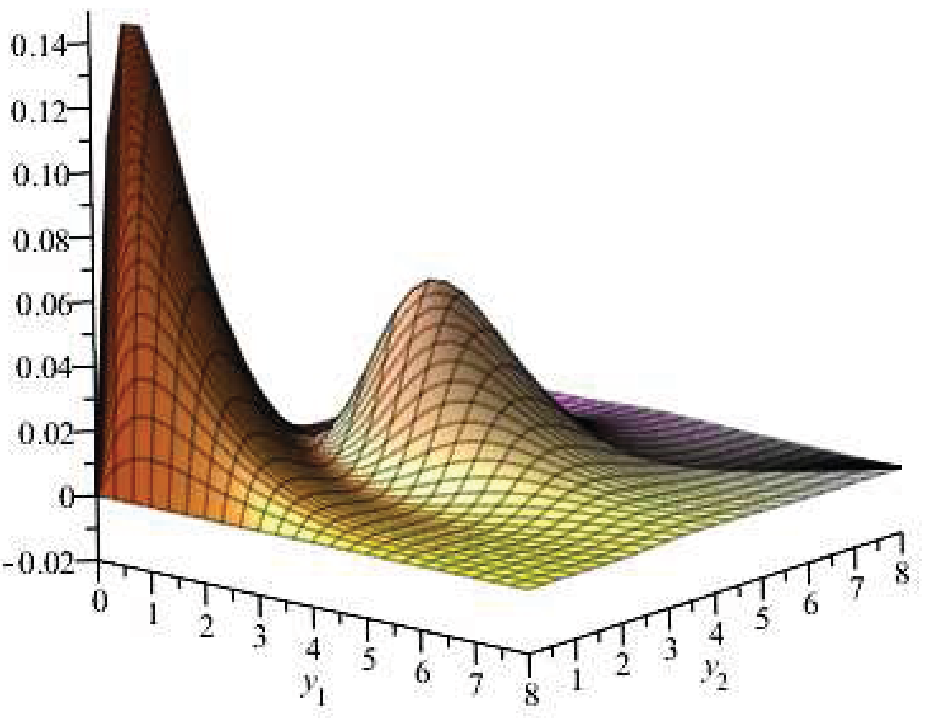}
\includegraphics[scale=0.7]{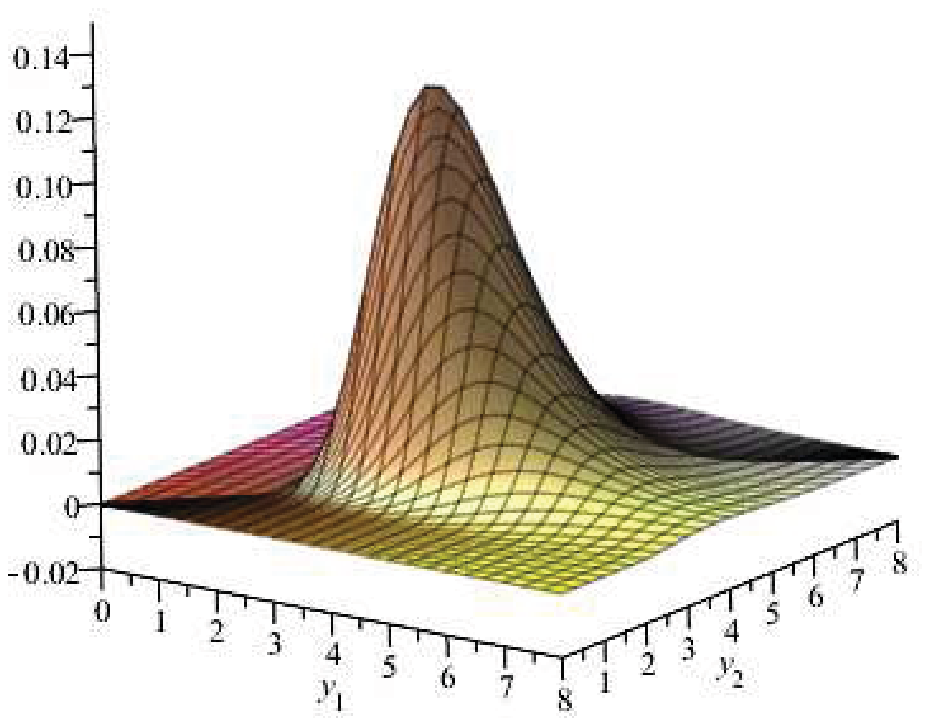}
\includegraphics[scale=0.7]{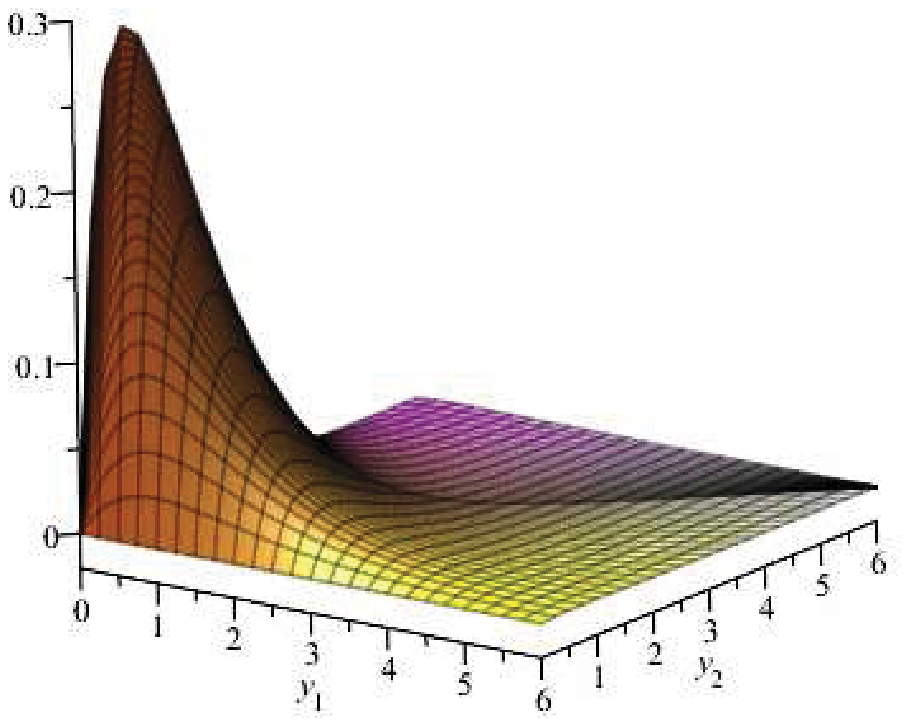}
\includegraphics[scale=0.7]{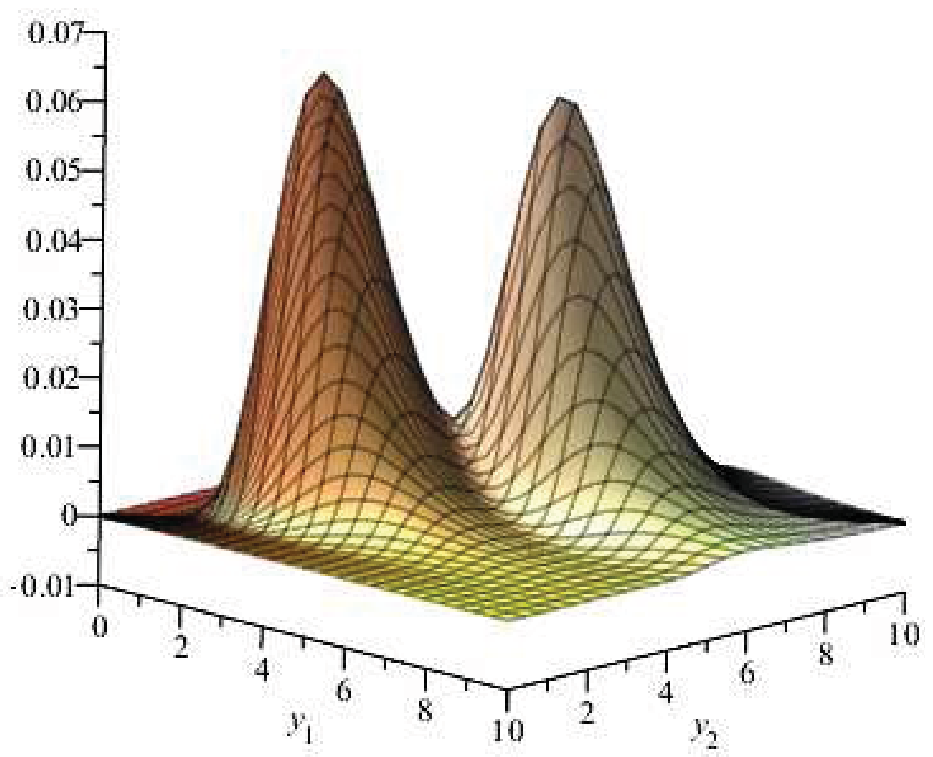}
\caption[]{The pdf of the BGEPS class of distribution for some values of parameters:
${\alpha }_1={\alpha }_2={\alpha }_3=1$, $\theta =1$ (left top),
${\alpha }_1={\alpha }_2={\alpha }_3=1$, $\theta =2$ (right top),
${\alpha }_1={\alpha }_2={\alpha }_3=1$, $\theta =0.5$ (left bottom),
${\alpha }_1={\alpha }_2={\alpha }_3=10$, $\theta =1$ (right bottom).}\label{fig.denh}
\end{figure}

\begin{remark} Since $\theta C''\left(\theta \right)+C'\left(\theta \right)=\sum^{\infty }_{n=1}{n^2a_n{\theta }^{n-1}}$ and $C'\left(\theta \right)=\sum^{\infty }_{n=1}{na_n{\theta }^{n-1}}$ , we have
\begin{eqnarray*}
&&f_1\left(y_1,y_2\right)=\sum^{\infty }_{n=1}{p_nf_{GE}\left(y_1;n{\alpha }_1+n{\alpha }_3,\lambda \right)f_{GE}\left(y_2;n{\alpha }_2,\lambda \right)},\\
&&f_2\left(y_1,y_2\right)=\sum^{\infty }_{n=1}{p_nf_{GE}\left(y_1;n{\alpha }_1,\lambda \right)f_{GE}\left(y_2;n{\alpha }_2+n{\alpha }_3,\lambda \right)},\\
&&f_0\left(y\right)=\frac{{\alpha }_3}{{\alpha }_1+{\alpha }_2+{\alpha }_3}\sum^{\infty }_{n=1}{p_nf_{GE}\left(y;n{\alpha }_1+n{\alpha }_2+n{\alpha }_3,\lambda \right)},
\end{eqnarray*}
where $p_n=P\left(N=n\right)=\frac{a_n{\theta }^n}{C\left(\theta \right)}$ and $f_{GE}\left(.;n\alpha ,\lambda \right)$ is the pdf of GE distribution with parameters $n\alpha $ and $\lambda $. Note that $f_{GE}(.;n\alpha ,\lambda )$ is the pdf of random variable $\max  (U_1,\dots ,U_n)$ where $U_i$'s are independent random variables from a GE distribution with parameters $\alpha $ and $\lambda$.

\end{remark}

\begin{corollary}
The joint pdf of the BGEPS distributions provided in Theorem \ref{thm.pdf} can be written as
\begin{equation}
f_{\boldsymbol Y}(y_1, y_2)=\frac{\alpha_{1}+\alpha_{2}}{\alpha_{1}+\alpha_{2}+\alpha_{3}}‎ g_a (y_1, y_2)+ \frac{\alpha_{3}}{\alpha_{1}+\alpha_{2}+\alpha_{3}}‎ g_s (y),
\end{equation}
where
\begin{eqnarray*}
&&g_a (y_1, y_2)=\frac{\alpha_{1}+\alpha_{2}+\alpha_{3}}{\alpha_{1}+\alpha_{2}} \left\{
\begin{array}{ll}
f_1 (y_1, y_2) & \ \ {\rm if}\ \ \  y‎_1< y_2‎ \\
f_2 (y_1, y_2) & \ \ {\rm if} \ \ \ y‎_2< y_1‎,
\end{array}%
\right.‎ \\
&&g_s (y)=\frac{\theta }{C\left(\theta \right)}f_{GE}\left(y;{\alpha }_1+{\alpha }_2+{\alpha }_3,\lambda \right)C'\left(\theta F_{GE}\left(y;{\alpha }_1+{\alpha }_2+{\alpha }_3,\lambda \right)\right) \ \ {\rm if} \ \ \ y_1=y_2=y,
\end{eqnarray*}
and $0$ otherwise. Clearly, $g_a (., .)$ is the absolute continuous part and $g_s (.)$ is the singular part. If $\alpha_{3}=0$, it does not have any singular part and it becomes an absolute continuous density function. Note that $g_s (.)$ is the pdf of GEPS distribution with parameters $\alpha_1+\alpha_2+\alpha_3$ and $\lambda$.

\end{corollary}

\begin{prop}
The conditional distribution of $Y_1$ given $Y_2‎\leq y_2‎$ is an absolute continuous distribution function with the following cdf:
\begin{equation*}
P(Y_1‎\leq y_1|Y_2‎\leq y_2)=\left\{
\begin{array}{ll}
\frac{C(\theta (1-e^{-\lambda y_1})^{\alpha_{1} + \alpha_{3}} (1-e^{-\lambda y_2})^{\alpha_{2}})}{C(\theta (1-e^{-\lambda y_2})^{\alpha_{2}+ \alpha_{3}})} &
 \ \  {\rm if} \ \ \  y‎_1< y_2‎ \\
\frac{C(\theta (1-e^{-\lambda y_2})^{\alpha_{2} + \alpha_{3}} (1-e^{-\lambda y_1})^{\alpha_{1}})}{C(\theta (1-e^{-\lambda y_2})^{\alpha_{2}+ \alpha_{3}})} &
  \ \  {\rm if}  \ \ \ y‎_2< y_1.
\end{array}%
\right.‎
\end{equation*}

\end{prop}

\begin{prop} The limiting distribution of BGEPS when $\theta \to 0^+$ is
\begin{eqnarray*}
{\mathop{\lim }_{\theta \to 0^+} F_{{\boldsymbol Y}}\left(y_1,y_2\right)}&=&{\mathop{\lim }_{\theta \to 0^+} \frac{C\left(\theta F_{{\boldsymbol X}}\left(y_1,y_2\right)\right)}{C\left(\theta \right)}}={\mathop{\lim }_{\theta \to 0^+} \frac{\sum^{\infty }_{n=1}{a_n{\theta }^n{\left(F_{{\boldsymbol X}}\left(y_1,y_2\right)\right)}^n}}{\sum^{\infty }_{n=1}{a_n{\theta }^n}}}\\
&=&\mathop{\lim }_{\theta \to 0^+} \frac{a_c{\left(F_{{\boldsymbol X}}\left(y_1,y_2\right)\right)}^c+\sum^{\infty }_{n=c+1}{a_n{\theta }^{n-c}{\left(F_{{\boldsymbol X}}\left(y_1,y_2\right)\right)}^n}}{a_c+\sum^{\infty }_{n=c+1}{a_n{\theta }^{n-c}}} \\
&=&{\left(F_{{\boldsymbol X}}\left(y_1,y_2\right)\right)}^c\\
&=&\left\{ \begin{array}{ll}
{\left(1-e^{-\lambda y_1}\right)}^{c({\alpha }_1+{\alpha }_3)}{\left(1-e^{-\lambda y_2}\right)}^{c{\alpha }_2} & \ \  {\rm if}\ \ \ y_1\le y_2 \\
{\left(1-e^{-\lambda y_1}\right)}^{c{\alpha }_1}{\left(1-e^{-\lambda y_2}\right)}^{c({\alpha }_2+{\alpha }_3)} & \ \  {\rm if}\ \ \  y_1>y_2, \end{array}
\right.
\end{eqnarray*}
which is the pdf of a BGE distribution with parameters $c\alpha_1$, $c\alpha_2$, $c\alpha_3$ and $\lambda$, where
$c=\min \{n\in {\mathbb N}: a_n>0\} $.
\end{prop}

 For the joint random variables $\left(Y_1,Y_2,N\right)$, consider equation \eqref{eq.pdfyn} when ${\boldsymbol X}$ has a BGE distribution. Since $(Y_1,Y_2|N=n)$ has a BGE with parameters $n{\alpha }_1$, $n{\alpha }_2$, $n{\alpha }_3$, and $\lambda $, the joint pdf of $\left(Y_1,Y_2,N\right)$ is
\[f_{Y_1,Y_2,N}\left(y_1,y_2,n\right)=\left\{ \begin{array}{ll}
\frac{a_n{\theta }^n}{C\left(\theta \right)}f_{1n}(y_1,y_2) & \ \ {\rm if} \ \ \ y_1<y_2 \\
\frac{a_n{\theta }^n}{C\left(\theta \right)}f_{2n}(y_1,y_2) & \ \ {\rm if} \ \ \ y_2<y_1 \\
\frac{a_n{\theta }^n}{C\left(\theta \right)}f_{0n}(y) &  \ \ {\rm if} \ \ \ y_1=y_2=y, \end{array}
\right.\]
where
\begin{eqnarray*}
&&f_{1n}(y_1,y_2)=n^2{\lambda }^2({\alpha }_1+{\alpha }_3){\alpha }_2e^{-\lambda y_1-\lambda y_2}(1-e^{-\lambda y_1})^{n({\alpha }_1+{\alpha }_3)-1}(1-e^{-\lambda y_2})^{n{\alpha }_2-1},\\
&&
f_{2n}(y_1,y_2)=n^2{\lambda }^2({\alpha }_2+{\alpha }_3){\alpha }_1e^{-\lambda y_1-\lambda y_2}(1-e^{-\lambda y_1})^{n{\alpha }_1-1}(1-e^{-\lambda y_2})^{n({\alpha }_2+{\alpha }_3)-1},\\
&&
f_{0n}(y)=n\lambda {\alpha }_3e^{-\lambda y}(1-e^{-\lambda y})^{n({\alpha }_1+{\alpha }_2+{\alpha }_3)-1}.
\end{eqnarray*}
The conditional probability mass function of $N$ given $Y_1=y_1$ and $Y_2=y_2$ is
\[f_{N|Y_1,Y_2}\left(n|y_1,y_2\right)=\left\{ \begin{array}{ll}
\frac{n^2a_n{\left(\theta  K_1\left(y_1,y_2\right)\right)}^{n-1}}{h_1(y_1,y_2)} & \ \ {\rm if} \ \ \ y_1<y_2 \\
\frac{n^2a_n{\left(\theta  K_2\left(y_1,y_2\right)\right)}^{n-1}}{h_2(y_1,y_2)} & \ \ {\rm if} \ \ \ y_2<y_1 \\
\frac{na_n{\left(\theta  K_0\left(y\right)\right)}^{n-1}}{h_0(y)} & \ \ {\rm if} \ \ \ y_1=y_2=y, \end{array}
\right.\]
where
\begin{eqnarray*}
h_1(y_1,y_2)&=&\theta F_{GE}\left(y_1;{\alpha }_1+{\alpha }_3,\lambda \right)F_{GE}\left(y_2;{\alpha }_2,\lambda \right)C''\left(\theta F_{GE}\left(y_1;{\alpha }_1+{\alpha }_3,\lambda \right)F_{GE}\left(y_2;{\alpha }_2,\lambda \right)\right)\\
&&+C'\left(\theta F_{GE}\left(y_1;{\alpha }_1+{\alpha }_3,\lambda \right)F_{GE}\left(y_2;{\alpha }_2,\lambda \right)\right),
\\
h_1(y_1,y_2)&=&\theta F_{GE}\left(y_1;{\alpha }_1,\lambda \right)F_{GE}\left(y_2;{\alpha }_2+{\alpha }_3,\lambda \right)C''\left(\theta F_{GE}\left(y_1;{\alpha }_1,\lambda \right)F_{GE}\left(y_2;{\alpha }_2+{\alpha }_3,\lambda \right)\right)\\
&&+C'\left(\theta F_{GE}\left(y_1;{\alpha }_1,\lambda \right)F_{GE}\left(y_2;{\alpha }_2+{\alpha }_3,\lambda \right)\right),
\\
h_0\left(y\right)&=&C'\left(\theta F_{GE}\left(y;{\alpha }_1+{\alpha }_2+{\alpha }_3,\lambda \right)\right),
\end{eqnarray*}
and
\begin{eqnarray*}
K_1\left(y_1,y_2\right)&=&(1-e^{-\lambda y_1})^{{\alpha }_1+{\alpha }_3}(1-e^{-\lambda y_2})^{{\alpha }_2}=F_{GE}\left(y_1;{\alpha }_1+{\alpha }_3,\lambda \right)F_{GE}\left(y_2;{\alpha }_2,\lambda \right),\\
K_2\left(y_1,y_2\right)&=&(1-e^{-\lambda y_1})^{{\alpha }_1}(1-e^{-\lambda y_2})^{{\alpha }_2+{\alpha }_3}=F_{GE}\left(y_1;{\alpha }_1,\lambda \right)F_{GE}\left(y_2;{\alpha }_2+{\alpha }_3,\lambda \right),\\
K_0\left(y\right)&=&(1-e^{-\lambda y})^{{\alpha }_1+{\alpha }_2+{\alpha }_3}{=F}_{GE}\left(y;{{\alpha }_1+\alpha }_2+{\alpha }_3,\lambda \right).
\end{eqnarray*}
Since ${\theta }^2C'''\left(\theta \right)+3\theta C''\left(\theta \right)+C'\left(\theta \right)=\sum^{\infty }_{n=1}{n^3a_n{\theta }^{n-1}}$, $\theta C''\left(\theta \right)+C'\left(\theta \right)=\sum^{\infty }_{n=1}{n^2a_n{\theta }^{n-1}}$ and $C'\left(\theta \right)=\sum^{\infty }_{n=1}{na_n{\theta }^{n-1}}$, therefore, we can obtain the conditional expectation of $N$ given $Y_1=y_1$ and $Y_2=y_2$ as
\begin{equation}\label{eq.ENY}
E\left(N|y_1,y_2\right)=\left\{ \begin{array}{ll}
\frac{R_1\left(y_1,y_2\right)}{h_1(y_1,y_2)} & \ \ {\rm if} \ \ \ y_1<y_2 \\
\frac{R_2\left(y_1,y_2\right)}{h_2(y_1,y_2)} & \ \ {\rm if} \ \ \ y_2>y_1 \\
\frac{\theta K_0\left(y\right)C''\left(\theta K_0\left(y\right)\right)+C'\left(\theta K_0\left(y\right)\right)}{h_0(y)} & \ \ {\rm if} \ \ \ y_1=y_2=y, \end{array}
\right.
\end{equation}
where
\begin{eqnarray*}
R_i(y_1,y_2)&=&{(\theta K_i(y_1,y_2))}^2C'''(\theta K_i(y_1,y_2))+3\theta K_i(y_1,y_2)C''(\theta K_i(y_1,y_2))\\
&&+C'(\theta K_i(y_1,y_2)),\ \ \  i=1,2.
\end{eqnarray*}

\section{Special Cases}
\label{sec.spe}
In this section, we consider some special cases of BGEPS distributions.

\subsection{Bivariate generalized exponential-geometric distribution}
When $C\left(\theta \right)=\frac{\theta }{1-\theta }$ ($0<\theta <1$), the power series distribution becomes the geometric distribution (truncated at zero). Therefore, the cdf of bivariate generalized exponential-geometric (BGEG) distribution is given by
\[F_{{\boldsymbol Y}}\left(y_1,y_2\right)=\left\{ \begin{array}{ll}
\frac{(1-\theta ){\left(1-e^{-\lambda y_1}\right)}^{{\alpha }_1+{\alpha }_3}{\left(1-e^{-\lambda y_2}\right)}^{{\alpha }_2}}{1-\theta {\left(1-e^{-\lambda y_1}\right)}^{{\alpha }_1+{\alpha }_3}{\left(1-e^{-\lambda y_2}\right)}^{{\alpha }_2}} & \ \ {\rm if}\ \ \  y_1\le y_2 \\
\frac{(1-\theta ){\left(1-e^{-\lambda y_1}\right)}^{{\alpha }_1}{\left(1-e^{-\lambda y_2}\right)}^{{\alpha }_2+{\alpha }_3}}{1-\theta {\left(1-e^{-\lambda y_1}\right)}^{{\alpha }_1}{\left(1-e^{-\lambda y_2}\right)}^{{\alpha }_2+{\alpha }_3}} & \ \ {\rm if}\ \ \ y_1>y_2, \end{array}
\right.\]
and its pdf is given in \eqref{eq.fBGEPS} with
\begin{eqnarray*}
f_1(y_1,y_2)&=&(1-\theta )f_{GE}(y_1;\alpha_1+\alpha_3,\lambda)f_{GE}(y_2;\alpha_2,\lambda)\frac{1+\theta F_{GE}(y_1;\alpha_1+\alpha_3,\lambda )F_{GE}(y_2;\alpha_2,\lambda)}{(1-\theta F_{GE}(y_1;{\alpha }_1+{\alpha }_3,\lambda )F_{GE}(y_2;{\alpha }_2,\lambda ))^3},\\
f_2(y_1,y_2)&=&(1-\theta )f_{GE}(y_1;{\alpha }_1,\lambda )f_{GE}(y_2;{\alpha }_2+{\alpha }_3,\lambda )\frac{1+\theta F_{GE}(y_1;{\alpha }_1,\lambda )F_{GE}(y_2;{\alpha }_2+{\alpha }_3,\lambda )}{{(1-\theta F_{GE}(y_1;{\alpha }_1,\lambda )F_{GE}(y_2;{\alpha }_2+{\alpha }_3,\lambda ))}^3},\\
f_0(y)&=&\frac{(1-\theta ){\alpha }_3f_{GE}(y;{\alpha }_1+{\alpha }_2+{\alpha }_3,\lambda )}{({\alpha }_1+{\alpha }_2+{\alpha }_3){(1-\theta F_{GE}(y;{\alpha }_1+{\alpha }_2+{\alpha }_3,\lambda ))}^2}.
\end{eqnarray*}

\begin{remark}
When ${\theta }^*=1-\theta $, we have
\[F_{{\boldsymbol Y}}\left(y_1,y_2\right)=\left\{ \begin{array}{ll}
\frac{{\theta }^*{\left(1-e^{-\lambda y_1}\right)}^{{\alpha }_1+{\alpha }_3}{\left(1-e^{-\lambda y_2}\right)}^{{\alpha }_2}}{1-(1-{\theta }^*){\left(1-e^{-\lambda y_1}\right)}^{{\alpha }_1+{\alpha }_3}{\left(1-e^{-\lambda y_2}\right)}^{{\alpha }_2}} & \ \ {\rm if}\ \ \ y_1\le y_2 \\
\frac{{\theta }^*{\left(1-e^{-\lambda y_1}\right)}^{{\alpha }_1}{\left(1-e^{-\lambda y_2}\right)}^{{\alpha }_2+{\alpha }_3}}{1-(1-{\theta }^*){\left(1-e^{-\lambda y_1}\right)}^{{\alpha }_1}{\left(1-e^{-\lambda y_2}\right)}^{{\alpha }_2+{\alpha }_3}} & \ \ {\rm if}\ \ \ y_1>y_2. \end{array}
\right.\]
It is also a cdf for all ${\theta }^*>0$
\citep[see][]{ma-ol-97}.
In fact, this is in Marshal-Olkin bivariate class of distributions.
\end{remark}

\subsection{Bivariate generalized exponential-Poisson distribution}

 When $a_n=\frac{1}{n!}$ and $C\left(\theta \right)=e^{\theta }-1$ ($\theta >0$), the power series distribution becomes the Poisson distribution (truncated at zero). Therefore, the cdf of bivariate generalized exponential- Poisson (BGEP) distribution is given by
\[F_{{\boldsymbol Y}}\left(y_1,y_2\right)=\left\{ \begin{array}{ll}
\frac{e^{\theta {\left(1-e^{-\lambda y_1}\right)}^{{\alpha }_1+{\alpha }_3}{\left(1-e^{-\lambda y_2}\right)}^{{\alpha }_2}}-1}{e^{\theta }-1} &
\ \ {\rm if}\ \ \ y_1\le y_2 \\
\frac{e^{\theta {\left(1-e^{-\lambda y_1}\right)}^{{\alpha }_1}{\left(1-e^{-\lambda y_2}\right)}^{{\alpha }_2+{\alpha }_3}}-1}{e^{\theta }-1} &
\ \ {\rm if}\ \ \ y_1>y_2, \end{array}
\right.\]
and its pdf is given in \eqref{eq.fBGEPS} with
\begin{eqnarray*}
f_1\left(y_1,y_2\right)&=&\theta f_{GE}\left(y_1;{\alpha }_1+{\alpha }_3,\lambda \right)f_{GE}\left(y_2;{\alpha }_2,\lambda \right)e^{\theta F_{GE}\left(y_1;{\alpha }_1+{\alpha }_3,\lambda \right)F_{GE}\left(y_2;{\alpha }_2,\lambda \right)-1}\\
&&\times
\left[\theta F_{GE}(y_1;{\alpha }_1+{\alpha }_3,\lambda )F_{GE}(y_2;{\alpha }_2,\lambda )+1\right],\\
f_2\left(y_1,y_2\right)&=&\theta f_{GE}\left(y_1;{\alpha }_1,\lambda \right)f_{GE}\left(y_2;{\alpha }_2+{\alpha }_3,\lambda \right)e^{\theta F_{GE}\left(y_1;{\alpha }_1,\lambda \right)F_{GE}\left(y_2;{\alpha }_2+{\alpha }_3,\lambda \right)-1}\\
&&\times
\left[\theta F_{GE}\left(y_1;{\alpha }_1,\lambda \right)F_{GE}\left(y_2;{\alpha }_2+{\alpha }_3,\lambda \right)+1\right],\\
f_0\left(y\right)&=&\frac{\theta {\alpha }_3}{\left({\alpha }_1+{\alpha }_2+{\alpha }_3\right)}f_{GE}\left(y;{\alpha }_1+{\alpha }_2+{\alpha }_3,\lambda \right)e^{\theta F_{GE}\left(y;{\alpha }_1+{\alpha }_2+{\alpha }_3,\lambda \right)-1}.
\end{eqnarray*}

\subsection{ Bivariate generalized exponential-binomial distribution}

When $a_n=\left(\genfrac{}{}{0pt}{}{k}{n}\right)$ and $C\left(\theta \right)={\left(\theta +1\right)}^k-1$ ($\theta >0$), where $k (n\le k)$ is the number of replicas, the power series distribution becomes the binomial distribution (truncated at zero). Therefore, the cdf of bivariate generalized exponential- binomial (BGEB) distribution is given by
\[F_{{\boldsymbol Y}}\left(y_1,y_2\right)=\left\{ \begin{array}{ll}
\frac{{\theta }^k{\left(1-e^{-\lambda y_1}\right)}^{k({\alpha }_1+{\alpha }_3)}{\left(1-e^{-\lambda y_2}\right)}^{k{\alpha }_2}-1}{{\left(\theta +1\right)}^k-1} & \ \ {\rm if}\ \ \ y_1\le y_2 \\
\frac{{\theta }^k{\left(1-e^{-\lambda y_1}\right)}^{k{\alpha }_1}{\left(1-e^{-\lambda y_2}\right)}^{k({\alpha }_2+{\alpha }_3)}-1}{{\left(\theta +1\right)}^k-1} & \ \ {\rm if}\ \ \ y_1>y_2, \end{array}
\right.\]
and its pdf is given in \eqref{eq.fBGEPS} with
\begin{eqnarray*}
f_1\left(y_1,y_2\right)&=&\frac{k\theta }{{\left(\theta +1\right)}^k-1}f_{GE}\left(y_1;{\alpha }_1+{\alpha }_3,\lambda \right)f_{GE}\left(y_2;{\alpha }_2,\lambda \right)\\
&&\times
{\left[\theta F_{GE}\left(y_1;{\alpha }_1+{\alpha }_3,\lambda \right)F_{GE}\left(y_2;{\alpha }_2,\lambda \right)+1\right]}^{k-2}\\
&&\times
\left[k\theta F_{GE}(y_1;{\alpha }_1+{\alpha }_3,\lambda )F_{GE}(y_2;{\alpha }_2,\lambda )+1\right],\\
f_2\left(y_1,y_2\right)&=&\frac{k\theta }{{\left(\theta +1\right)}^k-1}f_{GE}\left(y_1;{\alpha }_1,\lambda \right)f_{GE}\left(y_2;{\alpha }_2+{\alpha }_3,\lambda \right)\\
&&\times
{\left[\theta F_{GE}\left(y_1;{\alpha }_1,\lambda \right)F_{GE}\left(y_2;{\alpha }_2+{\alpha }_3,\lambda \right)+1\right]}^{k-2}\\
&&\times
\left[k\theta F_{GE}(y_1;{\alpha }_1,\lambda )F_{GE}(y_2;{\alpha }_{2+{\alpha }_3},\lambda )+1\right],\\
f_0(y)&=&\frac{k\theta {\alpha }_3f_{GE}(y;{\alpha }_1+{\alpha }_2+{\alpha }_3,\lambda )}{({(\theta +1)}^k-1)({\alpha }_1+{\alpha }_2+{\alpha }_3)}
{(\theta F_{GE}(y;{\alpha }_1+{\alpha }_2+{\alpha }_3,\lambda )+1)}^{k-1}
.
\end{eqnarray*}

\subsection{Bivariate generalized exponential-logarithmic distribution}

When $a_n=\frac{1}{n}$ and $C\left(\theta \right)=-{\log  \left(1-\theta \right)}$ ($0<\theta <1$), the power series distribution becomes the logarithmic distribution (truncated at zero). Therefore, the cdf of bivariate generalized exponential- logarithmic (BGEL) distribution is given by
\[F_{{\boldsymbol Y}}(y_1,y_2)=\left\{ \begin{array}{ll}
\frac{{\log  (1-\theta {(1-e^{-\lambda y_1})}^{{\alpha }_1+{\alpha }_3}{(1-e^{-\lambda y_2})}^{\alpha_2})}}{{\log  (1-\theta) }} &
\ \ {\rm if}\ \ \ y_1\le y_2 \\
\frac{{\log  (1-\theta {(1-e^{-\lambda y_1})}^{{\alpha }_1}{(1-e^{-\lambda y_2})}^{{\alpha }_2+{\alpha }_3}) }}{{\log  (1-\theta )}} &
\ \ {\rm if}\ \ \ y_1>y_2, \end{array}
\right.\]
and its pdf is given in \eqref{eq.fBGEPS} with
\begin{eqnarray*}
f_1\left(y_1,y_2\right)&=&\frac{-\theta f_{GE}\left(y_1;{\alpha }_1+{\alpha }_3,\lambda \right)f_{GE}\left(y_2;{\alpha }_2,\lambda \right)}{{\log  (1-\theta ) }{\left(1-\theta F_{GE}\left(y_1;{\alpha }_1+{\alpha }_3,\lambda \right)F_{GE}\left(y_2;{\alpha }_2,\lambda \right)\right)}^2},\\
f_2\left(y_1,y_2\right)&=&\frac{-\theta f_{GE}\left(y_1;{\alpha }_1,\lambda \right)f_{GE}\left(y_2;{\alpha }_2+{\alpha }_3,\lambda \right)}{{\log  (1-\theta ) }{\left(1-\theta F_{GE}\left(y_1;{\alpha }_1,\lambda \right)F_{GE}\left(y_2;{\alpha }_2+{\alpha }_3,\lambda \right)\right)}^2},\\
f_0\left(y\right)&=&\frac{-\theta {\alpha }_3f_{GE}\left(y;{\alpha }_1+{\alpha }_2+{\alpha }_3,\lambda \right)}{{\log  (1-\theta ) }\left({\alpha }_1+{\alpha }_2+{\alpha }_3\right)\left(1-\theta F_{GE}\left(y;{\alpha }_1+{\alpha }_2+{\alpha }_3,\lambda \right)\right)}.
\end{eqnarray*}

\subsection{ Bivariate generalized exponential- negative binomial distribution}

When $a_n=\binom{n-1}{k-1}$ and $C(\theta )=(\frac{\theta }{1-\theta} )^k$ ($0<\theta <1$), the power series distribution becomes the negative binomial distribution (truncated at zero). Therefore, the cdf of bivariate generalized exponential- negative binomial (BGENB) distribution is given by
\[F_{{\boldsymbol Y}}(y_1,y_2)=\left\{ \begin{array}{ll}
\frac{{(1-\theta )}^k{(1-e^{-\lambda y_1})}^{k{\alpha }_1+k{\alpha }_3}{(1-e^{-\lambda y_2})}^{k{\alpha }_2}}{{(1-\theta {(1-e^{-\lambda y_1})}^{{\alpha }_1+{\alpha }_3}{(1-e^{-\lambda y_2})}^{{\alpha }_2})}^k} &
\ \ {\rm if}\ \ \  y_1\le y_2 \\
\frac{{(1-\theta )}^k{(1-e^{-\lambda y_1})}^{k{\alpha }_1}{(1-e^{-\lambda y_2})}^{k{\alpha }_2+k{\alpha }_3}}{{(1-\theta {(1-e^{-\lambda y_1})}^{{\alpha }_1}{(1-e^{-\lambda y_2})}^{{\alpha }_2+{\alpha }_3})}^k} &
\ \ {\rm if}\ \ \ y_1>y_2, \end{array}
\right.\]
and its pdf is given in \eqref{eq.fBGEPS} with
\begin{eqnarray*}
f_1(y_1,y_2)&=&\frac{k{(1-\theta )}^kf_{GE}(y_1;{\alpha }_1+{\alpha }_3,\lambda )f_{GE}(y_2;{\alpha }_2,\lambda )}{{(1-\theta F_{GE}(y_1;{\alpha }_1+{\alpha }_3,\lambda )F_{GE}(y_2;{\alpha }_2,\lambda ))}^{k+2}}F^{k-1}_{GE}(y_1;{\alpha }_1+{\alpha }_3,\lambda )F^{k-1}_{GE}(y_2;{\alpha }_2,\lambda )\\
&&\times
 \left[k+\theta F_{GE}(y_1;{\alpha }_1+{\alpha }_3,\lambda )F_{GE}(y_2;{\alpha }_2,\lambda )\right],\\
f_2(y_1,y_2)&=&\frac{k{(1-\theta )}^kf_{GE}(y_1;{\alpha }_1,\lambda )f_{GE}(y_2;{\alpha }_2+{\alpha }_3,\lambda )}{{(1-\theta F_{GE}(y_1;{\alpha }_1,\lambda )F_{GE}(y_2;{\alpha }_2+{\alpha }_3,\lambda ))}^{k+2}}F^{k-1}_{GE}(y_1;{\alpha }_1,\lambda )F^{k-1}_{GE}(y_2;{\alpha }_2+{\alpha }_3,\lambda )\\
&&\times
 \left[k+\theta F_{GE}(y_1;{\alpha }_1,\lambda )F_{GE}(y_2;{\alpha }_2+{\alpha }_3,\lambda )\right],\\
f_0(y)&=&\frac{k{\alpha }_3{(1-\theta )}^kf_{GE}(y;{\alpha }_1+{\alpha }_2+{\alpha }_3,\lambda )F^{k-1}_{GE}(y;{\alpha }_1+{\alpha }_2+{\alpha }_3,\lambda )}{({\alpha }_1+{\alpha }_2+{\alpha }_3){(1-\theta F_{GE}(y;{\alpha }_1+{\alpha }_2+{\alpha }_3,\lambda ))}^{k+1}}.
\end{eqnarray*}

\section{Estimation}
\label{sec.est}

In this section, we consider the estimation of the unknown parameters of the BGEPS distributions. Let $\left(y_{11},y_{12}\right),\dots ,\left(y_{m1},y_{m2}\right)$ be an observed sample with size $m$ from BGEPS distributions with parameters ${\boldsymbol \Theta }=\left({\alpha }_1,{\alpha }_2,{\alpha }_3,\lambda ,\theta \right)'$. Also, consider
\[I_0=\left\{i:y_{1i}=y_{2i}=y_i\right\},\ \ \ \ \ \ \ I_1=\left\{i:y_{1i}<y_{2i}\right\},\ \ \ \ \ \ I_2=\left\{i:y_{1i}>y_{2i}\right\},\]
and
\[m_0=\left|I_0\right|,\ \ \ \ \ \ m_1=\left|I_1\right|,\ \ \ \ \ \ m_2=\left|I_2\right|,\ \ \ \ \ \ m=m_0+m_1+m_2.\]
Therefore, the log-likelihood function can be written as
\begin{equation}\label{eq.lik}
\ell \left({\boldsymbol \Theta }\right)=\sum_{i\in I_0}{{\log  \left(f_0\left(y_i\right)\right) }}+\sum_{i\in I_1}{{\log  \left(f_1\left(y_{1i},y_{2i}\right)\right)}}+\sum_{i\in I_2}{{\log  \left(f_{{\rm 2}}\left(y_{1i},y_{2i}\right)\right)}},
\end{equation}
where  $f_1$, $f_2$ and $f_0$ are given in \eqref{eq.f1}, \eqref{eq.f2} and \eqref{eq.f0}, respectively.  We can obtain the MLE's of the parameters by maximizing $\ell \left({\boldsymbol \Theta }\right)$ in \eqref{eq.lik} with respect to the unknown parameters. This is clearly a five-dimensional problem. However, no explicit expressions are available for the MLE's. We need to solve five non-linear equations simultaneously, which may not be very simple. The maximization can be performed using a command like the {\it nlminb} routine in the R software
\citep{rdev-14}.
But, it is related to initial guesses. Therefore, we present an expectation-maximization (EM) algorithm similar to \cite{ku-de-09-EM}  to find the MLE's of parameters.

For given $n$, consider that independent random variables $\{Z_i|N=n\}$, $i=1,2,3$ have the GE distribution with parameters $n{\alpha }_i$ and $\lambda $. It is well-known that
\[\{Y_1|N=n\}=\{{\max  (Z_1,Z_3)}|N=n\},\ \ \ \ \ \ \ \ \ \{Y_2|N=n\}=\{{\max  (Z_2,Z_3) }|N=n\}.\]

Assumed that for the bivariate random vector $ (Y_1,Y_2)$, there is an associated random vectors
\[{\Lambda }_1=\left\{ \begin{array}{cc}
0 & Y_1=Z_1 \\
1 & Y_1=Z_2 \end{array}
\right.\ \ \ \ \ \ \ \ {\rm and}\ \ \ \ \ \ \ {\Lambda }_2=\left\{ \begin{array}{cc}
0 & Y_2=Z_1 \\
1 & Y_2=Z_3. \end{array}
\right.\ \]

Note that if $Y_1=Y_2$, then $\Lambda_1={\Lambda }_2=0$. But if $Y_1<Y_2$ or $Y_1>Y_2$, then $({\Lambda }_1,{\Lambda }_2)$ is missing. If $\left(Y_1,Y_2\right)\in I_1$ then the possible values of $({\Lambda }_1,{\Lambda }_2)$ are $\left(1,0\right)$ or $(1,1)$, and If $\left(Y_1,Y_2\right)\in I_2$ then the possible values of $({\Lambda }_1,{\Lambda }_2)$ are $\left(0,1\right)$ or $(1,1)$ with non-zero probabilities.

We form the conditional `pseudo' log-likelihood function, conditioning on $N$, and then replace $N$ by $E(N|Y_1,Y_2)$. In the E-step of the EM-algorithm, we treat it as complete observation when they belong to $I_0$. If the observation belong to $I_1$, we form the `pseudo' log-likelihood function by fractioning $(y_1,y_2)$ to two partially complete `pseudo' observations of the form $(y_1,y_2,u_1\left({\boldsymbol \Theta }\right))$  and $(y_1,y_2,u_2\left({\boldsymbol \Theta }\right))$, where $u_1\left({\boldsymbol \Theta }\right)$ and $u_2\left({\boldsymbol \Theta }\right)$ are the conditional probabilities that $({\Lambda }_1,{\Lambda }_2)$ takes values $\left(1,0\right)$ and $(1,1)$, respectively. Since
\begin{eqnarray*}
&&P\left(Z_3<Z_1<Z_2|N=n\right)=\frac{{\alpha }_1{\alpha }_2}{\left({\alpha }_1+{\alpha }_3\right)\left({\alpha }_1+{\alpha }_2+{\alpha }_3\right)},\\
&&P\left(Z_1<Z_3<Z_2|N=n\right)=\frac{{\alpha }_2{\alpha }_3}{\left({\alpha }_1+{\alpha }_3\right)\left({\alpha }_1+{\alpha }_2+{\alpha }_3\right)},
\end{eqnarray*}
therefore,
\begin{equation}
u_1\left({\boldsymbol \Theta }\right)=\frac{{\alpha }_1}{{\alpha }_1+{\alpha }_3},\ \ \ \ \ \ \ u_2\left({\boldsymbol \Theta }\right)=\frac{{\alpha }_3}{{\alpha }_1+{\alpha }_3}.
\end{equation}

Similarly, If the observation belong to $I_2$, we form the `pseudo' log-likelihood function of the from $\left(y_1,y_2,v_1\left({\boldsymbol \Theta }\right)\right)$ and $\left(y_1,y_2,v_2\left({\boldsymbol \Theta }\right)\right)$, where $v_1\left({\boldsymbol \Theta }\right)$ and $v_2\left({\boldsymbol \Theta }\right)$ are the conditional probabilities that $({\Lambda }_1,{\Lambda }_2)$ takes values $\left(0,1\right)$ and $(1,1)$, respectively. Therefore,
\begin{equation}
v_1\left({\boldsymbol \Theta }\right)=\frac{{\alpha }_2}{{\alpha }_2+{\alpha }_3},\ \ \ \ \ \ \ v_2\left({\boldsymbol \Theta }\right)=\frac{{\alpha }_3}{{\alpha }_2+{\alpha }_3}.
\end{equation}
For brevity, we write  $u_1\left({\boldsymbol \Theta }\right)$, $u_2\left({\boldsymbol \Theta }\right)$, $v_1\left({\boldsymbol \Theta }\right)$, $v_2\left({\boldsymbol \Theta }\right)$ as $u_1$, $u_2$, $v_1$, $v_2$, respectively.

\noindent\textbf{E-step:} Consider $b_i=E(N|y_{1i},y_{2i},{\boldsymbol \Theta })$. The log-likelihood function without the additive constant can be written as follows:
\begin{eqnarray}\label{eq.lps}
{\ell }_{{\rm pseudo}}({\boldsymbol \Theta })&=&{\log  (\theta ) }\sum^m_{i=1}{b_i}-m{\log (C(\theta )) }+{(m}_0+2m_1+2m_2){\log  (\lambda ) }
\nonumber\\
&&+(m_1u_1+m_2 )\log ({\alpha}_1)+(m_1+m_2v_1){\log  ({\alpha }_2) }\nonumber\\
&&
+(m_0+m_1u_2+m_2v_2){\log  ({\alpha }_3)}-\lambda(\sum_{i\in I_0}{y_i}+\sum_{i\in I_1\cup I_2}{(y_{1i}+y_{2i})})
\\
&&+{\alpha }_1(\sum_{i\in I_0}{b_iQ(y_i)}+\sum_{i\in I_1\cup I_2}{b_i Q( y_{1i}) })
+{\alpha }_2(\sum_{i\in I_0}b_iQ(y_i) +\sum_{i\in I_1\cup I_2}b_i Q(y_{2i}) )\nonumber\\
&&+{\alpha }_3(\sum_{i\in I_0}{b_iQ(y_i)}+\sum_{i\in I_1}{b_i Q(y_{1i})}+\sum_{i\in I_2}{b_i Q(y_{2i})})\nonumber\\
&&-\sum_{i\in I_0}Q(y_i)-\sum_{i\in I_1\cup I_2}Q(y_{1i})-\sum_{i\in I_1\cup I_2}{b_iQ(y_{2i})},\nonumber
\end{eqnarray}
where $Q(y)=\log  (1-e^{-\lambda y})$.

\noindent\textbf{M-step:} At this step, ${\ell }_{{\rm pseudo}}({\boldsymbol \Theta })$ is maximized with respect to ${\alpha }_1,{\alpha }_2,{\alpha }_3,\lambda $ and $\theta $. For fixed $\lambda $, the maximization occurs at
\begin{eqnarray}
{\hat{\alpha }}_1(\lambda )&=&\frac{-(m_1u_1+m_2)}{\sum_{i\in I_0}{b_iQ(y_i)}+\sum_{i\in I_1\cup I_2}{b_iQ(y_{1i})}}, \label{eq.al1}\\
{\hat{\alpha }}_2(\lambda )&=&\frac{-(m_1+m_2v_1)}{\sum_{i\in I_0}{b_iQ(y_i)}+\sum_{i\in I_1\cup I_2}{b_iQ(y_{2i})}},\label{eq.al2}\\
{\hat{\alpha }}_3(\lambda )&=&\frac{-(m_0+m_1u_2+m_2v_2)}{\sum_{i\in I_0}b_iQ(y_i)+\sum_{i\in I_1}b_iQ(y_{1i})
+\sum_{i\in I_2}b_iQ(y_{2i})}, \ \ \ \ \ \label{eq.al3}
\end{eqnarray}
and solving the following non-linear equation with respect to $\theta $:
\begin{equation}\label{eq.thh}
\frac{\theta C'\left(\theta \right)}{C\left(\theta \right)}=\frac{\sum^m_{i=1}{b_i}}{m}.
\end{equation}

\begin{remark}
When $C\left(\theta \right)=\frac{\theta }{1-\theta }$, then the solution of equation in \eqref{eq.thh} is $\hat{\theta }=1-\frac{m}{\sum^m_{i=1}{b_i}}$.
\end{remark}

\begin{remark}
We do not need to solve the equation in \eqref{eq.thh}, when $C\left(\theta \right)=\theta $. In fact, the BGEPS distribution reduces to the BGE distribution.
\end{remark}

Finally, $\hat{\lambda }$ can be obtained as a solution of the following equation:
\begin{equation}\label{eq.lamh}
\frac{m_0+2m_1+2m_2}{{\rm g}(\lambda )}=\lambda ,
\end{equation}
where
\begin{eqnarray*}
{\rm g}(\lambda )&=&\sum_{i\in I_0}{y_i}+\sum_{i\in I_1\cup I_2}{(y_{1i}+y_{2i})}-\sum_{i\in I_0}{(b_i{\hat{\alpha }}_1+b_i{\hat{\alpha }}_2+b_i{\hat{\alpha }}_3-1)\frac{y_ie^{-\lambda y_i}}{1-e^{-\lambda y_i}}}\\
&&-\sum_{i\in I_1}{\left(b_i{\hat{\alpha }}_1+b_i{\hat{\alpha }}_3-1\right)\frac{y_{1i}e^{-\lambda y_{1i}}}{1-e^{-\lambda y_{1i}}}}-\sum_{i\in I_1}{\left(b_i{\hat{\alpha }}_2-1\right)\frac{y_{2i}e^{-\lambda y_{2i}}}{1-e^{-\lambda y_{2i}}}}\\
&&
-\sum_{i\in I_2}{\left(b_i{\hat{\alpha }}_1-1\right)\frac{y_{1i}e^{-\lambda y_{1i}}}{1-e^{-\lambda y_{1i}}}}-\sum_{i\in I_2}{\left(b_i{\hat{\alpha }}_2+b_i{\hat{\alpha }}_3-1\right)\frac{y_{2i}e^{-\lambda y_{2i}}}{1-e^{-\lambda y_{2i}}}}.
\end{eqnarray*}

The following steps can be used to compute the MLE's of the parameters via the EM algorithm:\\
\textbf{Step 1}: Take some initial value of ${\boldsymbol \Theta }$, say
${\boldsymbol \Theta }^{(0)} =(\alpha^{(0)}_1,\alpha^{(0)}_2,\alpha^{(0)}_3,\lambda^{(0)},\theta ^{(0)})'$.\\
\textbf{Step 2}: compute $b_i=E(N| y_{1i},y_{2i};{\boldsymbol \Theta }^{(0)})$\\
\textbf{Step 3: }Compute $u_1$, $u_2$, $v_1$, and $v_2$.\\
\textbf{Step 4}: Find $\hat{\lambda }$ by solving the equation \eqref{eq.lamh}, say ${\hat{\lambda }}^{(1)}$.\\
\textbf{Step 5}: Compute ${\hat{\alpha }}^{(1)}_i={\hat{\alpha }}_i({\hat{\lambda }}^{(1)})$, $i=1,2,3$ from \eqref{eq.al1}-\eqref{eq.al3}.\\
\textbf{Step 6}: Find $\hat{\theta }$ by solving the equation \eqref{eq.thh}, say ${\hat{\theta }}^{(1)}$.\\
\textbf{Step 7}: Replace ${{\boldsymbol \Theta }}^{{\rm (0)}}$ by
${\boldsymbol \Theta }^{(1)} =(\alpha^{(1)}_1,\alpha^{(1)}_2,\alpha^{(1)}_3,\lambda^{(1)},\theta^{(1)})$, go back to step 1 and continue the process until convergence take place.

\begin{figure}[ht]
\centering
\includegraphics[scale=0.5]{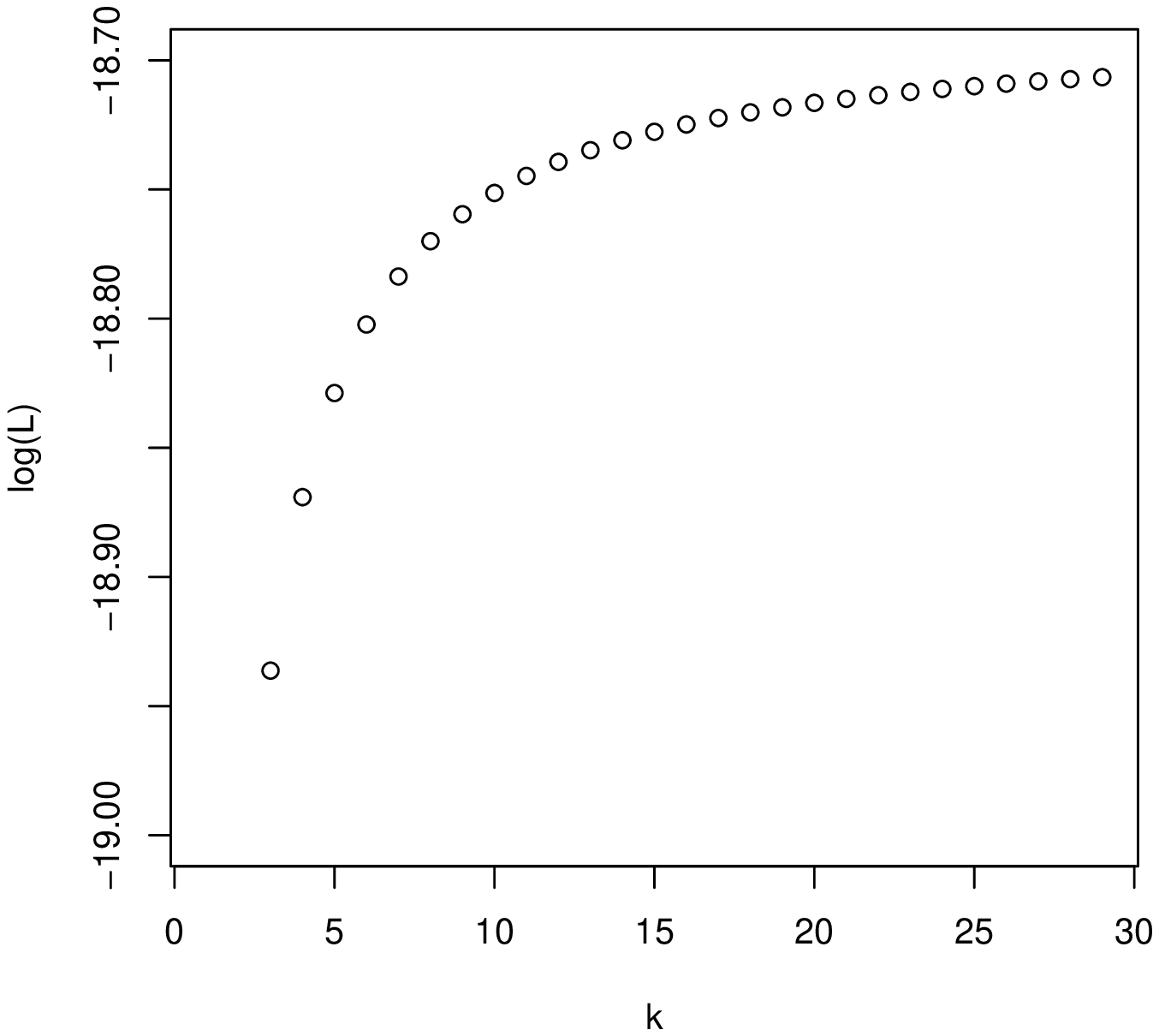}
\includegraphics[scale=0.5]{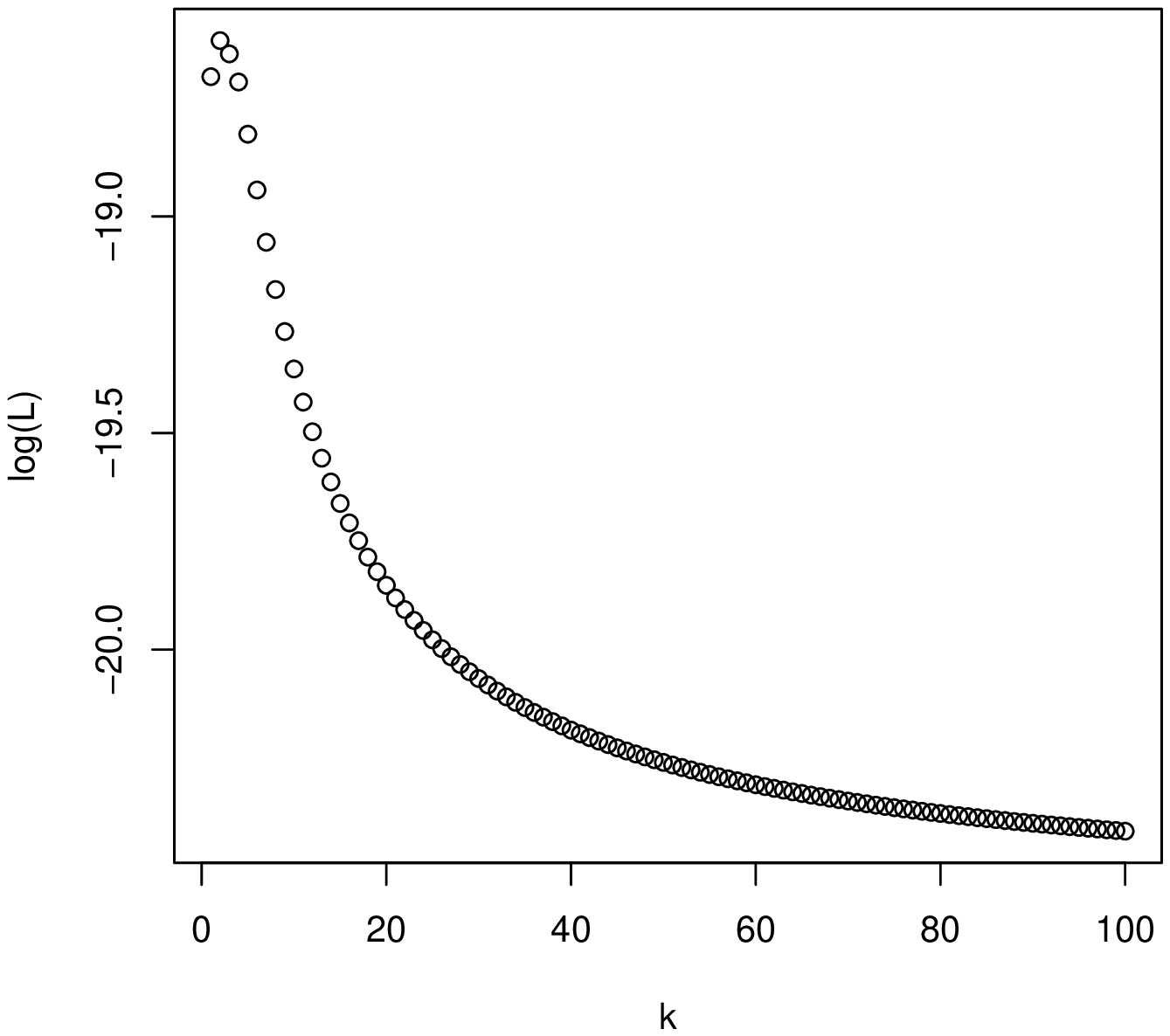}
\vspace{-0.5cm}
\caption{The log-likelihood values of BGEB (left) and BGENB (right) distributions for different values of $k$. }
\label{fig.liks}
\end{figure}

\section{A real example}
\label{sec.exa}
The data set is given from
\cite{meintanis-07}
and is obtained from the group stage of the UEFA Champion's League for the years 2004-05 and 2005-2006. In addition,
\cite{ku-gu-09-BGE}
and
\cite{ku-de-09-EM}
analyzed this data by using BGE and Marshall-Olkin bivariate Weibull distributions, respectively.
The data represent the football (soccer) data where at least one goal scored by the home team and at least one goal scored directly from a kick goal (like penalty kick, foul kick or any other direct kick) by any team have been considered. Here $Y_1$ represents the time in minutes of the first kick goal scored by any team and $Y_2$ represents the first goal of any type scored by the home team.

We divided all the data by 100. Then six special cases of BGEPS distributions are considered:  BGE, BGEG, BGEP, BGEB, BGENB, and BGEL. Using the proposed EM algorithm, these models are fitted to the bivariate data set, and the MLE's and their corresponding log-likelihood values are calculated. The standard errors (se) based on the observed information matrix are obtained. The results are given in Table \ref{tab.ex1}.

\begin{table}[ht]

\caption{The MLE's, log-likelihood, AIC, AICC, BIC, K-S, and LRT statistics for six sub-models of BGEPS distribution.}\label{tab.ex1}
\begin{center}
\begin{tabular}{|c|c|c|c|c|c|c|} \hline
 & \multicolumn{6}{|c|}{Distribution} \\ \hline
Statistic & BGE & BGEG & BGEP & BGEB\newline ($k=30)$ & BGENB\newline ($k=2)$ & BGEL \\ \hline
${\hat{\alpha }}_1$ & 1.4452 & 0.9964 & 0.5644 & 0.5980 & 0.2538 & 1.1871 \\
(s.e.) & (0.4160) & (0.4938) & (0.5758) & (0.5442) & (0.3340) & (0.4140) \\ \hline
${\hat{\alpha }}_2$ & 0.4681 & 0.3047 & 0.1676 & 0.1780 & 0.0755 & 0.3646 \\
(s.e.) & (0.1879) & (0.1825) & (0.1792) & (0.1717) & (0.1044) & (0.1659) \\ \hline
${\hat{\alpha }}_3$ & 1.1704 & 0.7205 & 0.4009 & 0.4264 & 0.1792 & 0.8615 \\
(s.e.) & (0.2866) & (0.3811) & (0.4053) & (0.3843) & (0.2443) & (0.2929) \\ \hline
$\hat{\lambda }$ & 3.8994 & 5.0600 & 4.6565 & 4.6252 & 4.8877 & 5.3466 \\
(s.e.) & (0.5603) & (0.8472) & (0.6698) & (0.6620) & (0.7579) & (1.0009) \\ \hline
$\hat{\theta }$ & --- & 0.6932 & 4.1980 & 0.1484 & 0.8212 & 0.8844 \\
(s.e.) & --- & (0.2318) & (3.7768) & (0.1351) & (0.2739) & (0.1250) \\ \hline
${\log  (\ell )\ }$ & -20.5904 & -18.6775 & -18.6864 & -18.7065 & -18.5941 & -18.4599 \\
AIC  & 49.1809 & 47.3549 & 47.3728 & 47.4130 & 47.1882 & 46.9198 \\
AICC  & 51.0856 & 49.2904 & 49.3083 & 49.3485 & 49.1237 & 48.8553 \\
BIC & 54.2133 & 55.4095 & 55.4274 & 55.4676 & 55.2428 & 54.9744 \\ \hline
K-S ($Y_1$) & 0.1034 & 0.0952 & 0.0865 & 0.0871 & 0.0837 & 0.1011 \\
(p-value) & (0.8240) & (0.8906) & (0.9448) & (0.9418) & (0.9576) & (0.8435) \\ \hline
K-S ($Y_2$) & 0.1001 & 0.0900 & 0.0957 & 0.0961 & 0.0920 & 0.0843 \\
(p-value) & (0.8527) & (0.9255) & (0.8870) & (0.8844) & (0.9132) & (0.9551) \\ \hline
K-S (${\max  (Y_1,Y_2)\ }$) & 0.1431 & 0.1405 & 0.1329 & 0.1331 & 0.1344 & 0.1518 \\
(p-value) & (0.4344) & (0.4587) & (0.5301) & (0.5289) & (0.5161) & (0.3613) \\ \hline
LRT & --- & 3.8258 & 3.8079 & 3.7677 & 3.9925 & 4.2609 \\
(p-value) & --- & (0.0504) & (0.0510) & (0.0522) & (0.0457) & (0.0389) \\ \hline
\end{tabular}

\end{center}
\end{table}

For choosing the best value for $k$ in BGEB distribution, we considered $k=2,\dots ,30$, and obtained the corresponding log-likelihood values for all models. The results are given in Figure \ref{fig.liks}. It can be concluded that the log-likelihood values increase when $k$ increases. But there is negligible variation for large $k$. Therefore, the results for $k=30$ are given in Table \ref{tab.ex1}.

Similarly, for choosing the best value for $k$ in BGENB distribution, we consider $k=1,\dots ,100$. The results are given in Figure \ref{fig.liks}. It can be concluded that the largest log-likelihood value occurs in the case of $k=2$. Therefore, we present the results for $k=2$ in Table \ref{tab.ex1}.

For each fitted model, the Akaike Information Criterion (AIC), the corrected Akaike information criterion (AICC) and the Bayesian information criterion (BIC) are calculated. We also obtain the Kolmogorov-Smirnov (K-S) distances between the fitted distribution and the empirical distribution function and the corresponding p-values (in brackets) for $Y_1$, $Y_2$  and $ \max (Y_1,Y_2 )$.

Finally, we make use the likelihood ratio test (LRT) for testing the BGE against other models. The statistics and the corresponding p-values are given in Table \ref{tab.ex1}.

\bibliographystyle{apa}

\newpage

\end{document}